\documentclass[letterpaper, 10 pt, journal, twoside]{IEEEtran}
\IEEEoverridecommandlockouts

\usepackage{graphicx}
\usepackage{stfloats}
\usepackage{url}
\usepackage{hyperref}
 \usepackage{amsmath,amssymb,amsfonts,subfigure,subcaption}
 \usepackage{xcolor}
 \usepackage{caption}
 \usepackage{floatrow}

\usepackage{amsthm}

\newtheorem{theorem}{Theorem}

\newtheorem{assumption}{Assumption}
\newtheorem{remark}{Remark}
\newtheorem{definition}{Definition}

\newtheorem{problem}{Problem}

\begin{document}
\title{Safety in Admittance Control using Reference Trajectory Shaping}

\author {Chayan Kumar Paul, Bhabani Shankar Dey*, and Indra Narayan Kar
\thanks{$^*$Corresponding author: Bhabani Shankar Dey}
\thanks{$^{1}$Chayan Kumar Paul is a PhD student in the Department of Electrical Engineering, Indian Institute of Technology Delhi, Hauz Khas, New Delhi, 110016, India
{\tt\small chayanpaul007@gmail.com}}
\thanks{$^{2}$Bhabani Shankar Dey is with Centre for Cyber-Physical Systems, Indian Institute of Science, Bengaluru, India
        {\tt\small bhabanishankar440@gmail.com}
        }%
\thanks{$^{3}$Indra Narayan Kar is a faculty in the Department of Electrical Engineering, Indian Institute of Technology Delhi, Hauz Khas, New Delhi, 110016, India
{\tt\small ink@ee.iitd.ac.in}}}

\maketitle

\begin{abstract}                
This paper presents a switched model-reference admittance control framework to achieve safe and compliant human–robot collaboration through reference trajectory shaping. The proposed method generates variable admittance parameters according to task compliance and task-space safety requirements. Additionally, a disturbance bound  is incorporated to enhance robustness against disturbances. Safety guarantees are explicitly established by integrating invariance control, ensuring that the reference trajectory remains within the admissible region. Stability of the switched system is analyzed using a common quadratic Lyapunov function, which confirms asymptotic convergence of the tracking error. The effectiveness of the approach is demonstrated through simulations on a two-link manipulator and comparison with existing methods are also presented. Furthermore, real-time implementation on a single-link manipulator validates the practical feasibility of the controller, highlighting its ability to achieve both compliance and safety in physical interaction scenarios.
\end{abstract}



\section{Introduction}
Industry 5.0 marks a new era in manufacturing, where robots are no longer confined to repetitive tasks in isolated spaces but are evolving into intelligent partners that share workspaces with humans, ensuring safety, comfort, and productivity. Its core vision is seamless human–robot collaboration, combining human cognitive strengths with robotic precision, speed, and flexibility. 
Impedance and admittance control strategies~\cite{hogan,abu2020variable,song2019tutorial} are widely adopted for this purpose, modeling the interaction as a virtual spring–mass–damper system. In particular, admittance control generates a reference trajectory based on the applied force, enabling smooth and compliant cooperation.
Based on the reference trajectories generated through admittance control, a variety of controllers have been developed for rehabilitation~\cite{Pehlivan2016Minimal,Yu2015Human}, exoskeletons~\cite{Li2018Asymmetric}, industrial applications~\cite{He2017Model,mariotti2019admittance} and many more.\\
In physical interaction scenarios, however, two key requirements dominate: compliance and safety. Compliance is typically ensured by tuning the virtual spring–mass–damper (or admittance) parameters according to the task, whereas safety requires guaranteeing that the system states remain within prescribed task-space constraints. A system is regarded as safe if its trajectories evolve entirely within a predefined admissible region. Existing approaches for safety-critical control in robotics generally fall into two categories: (i) path planning, which generates collision-free trajectories to ensure safety, and (ii) safety filters, which adjust controller inputs to enforce safety constraints. Safety through path planning has been applied primarily to mobile robots~\cite{al2003efficient}, but such heuristic-based approaches are often computationally expensive and time-consuming. In contrast, the development of safety filters has gained significant popularity due to their relative ease of implementation. Representative examples include Model Predictive Control (MPC) \cite{nubert2020safe}, Control Barrier Functions (CBF) \cite{cbf}, Barrier Lyapunov Functions (BLF) \cite{ablf}, and invariance control \cite{kimmel2017invariance}. Each of these approaches exhibits distinct advantages and limitations. MPC and CBF provide rigorous constraint satisfaction, but both require solving optimization problems online at every sampling instant, making real-time implementation computationally demanding, particularly in the presence of multiple constraints. BLF-based methods offer algebraic solutions; however, the control effort increases sharply near constraint boundaries, which can render the input infeasible~\cite{ablf}. Prescribed-performance control (PPC) \cite{meng2024adaptive} guarantees that the tracking error remains within a prespecified bound, which can also serve as a safety certificate. Nevertheless, implementing PPC in real time often requires very large control inputs, exceeding actuator limits. Moreover, both BLF and PPC lack mechanisms to handle constraint violations once they occur, potentially leading to instability or undesired behavior. Invariance control, in contrast, ensures task-space safety with feasible control effort by employing switching laws, but at the cost of oscillations in the system output\cite{wolff2007continuous}. Originally introduced for nonlinear control-affine SISO systems and later extended to multiple constraints, invariance control has also been applied to safe human–robot interaction~\cite{kimmel2017invariance}. However, this implementation exhibits significant chattering in the control input.\\
Most existing strategies implicitly assume that the reference trajectory always remains within the safe region. In practice, however, even with accurate system models, human-applied interaction forces can drive the system toward constraint violations. Since such forces are not treated as disturbances and cannot be entirely rejected, a natural solution is to adapt admittance parameters to preserve safety. Building on this idea,~\cite{chayansafety} proposed a switched model-reference admittance control framework in which reference parameters adapt dynamically according to task compliance. However, that approach did not provide formal safety guarantees. In contrast, the method proposed here explicitly integrates invariance control with a switching reference model, thereby ensuring safety. The resulting switching law constrains the generated reference trajectory within the admissible task-space set, simultaneously achieving compliance and safety with feasible control effort. Furthermore, a Lyapunov-based analysis establishes asymptotic convergence of the switched system, providing a rigorous theoretical foundation for the approach. The key contributions of this article are summarized as follows:
\begin{itemize}
    \item An adaptive switched reference model is proposed within a variable admittance control framework to simultaneously ensure compliance and constraint satisfaction. 
    \item Practical safety guarantees are established by constructing a control invariance set, ensuring that the reference trajectory remains within the admissible task-space.
    \item To enhance robustness against disturbances and measurement noise, a bound on the error dynamics is incorporated within the switching law, thereby preserving stability and performance.  
\end{itemize}

Furthermore, the effectiveness of the proposed method is validated through comparisons with BLF, PPC, and invariance control schemes. Also, a hardware implementation on a single-link manipulator demonstrates the practical applicability of the approach.  

\textbf{Notations:}  $\mathbb{R}$ is the set of real numbers. Low-order time derivatives are denoted by \( \dot{x} = \frac{dx}{dt} \),  \( \ddot{x} = \frac{d^2x}{dt^2} \) and the $k^{th}$-order time derivative is represented by \( x^{(k)} = \frac{d^k x}{dt^k} \) where \(k>2\). The notation $A \preceq B$ indicates that $A_i \leq B_i$ for all components $i$ of the vector $A$. For a vector $x \in \mathbb{R}^n$, the notation $|x|$ denotes the element-wise absolute value, i.e.,\(
|x| = \begin{bmatrix} |x_1| & |x_2| & \cdots & |x_n| \end{bmatrix}^\top . 
\) Also, $0$ and $I$ represent the zero and identity matrix of appropriate dimensions, respectively.

\section{Preliminaries and Problem Formulation}
\subsection{Set Invariance to ensure safety} \label{section:set-inva}

Invariance control~\cite{2008invariance} provides a systematic framework to guarantee safety in systems with state and input constraints. The core idea is to combine a nominal controller $u_{nom}$, designed for trajectory tracking within a safe set, with a corrective controller $u_{cor}$ that overrides the nominal action near the boundary of the safe set to preclude constraint violations. In this way, full actuator authority is leveraged to steer the state back into the admissible region and maintain safety.  

A central component of this approach is the notion of invariance set, which is a subset of the state space where safety needs to be guaranteed under suitable control. Formally, the positively controlled invariance set is defined as follows.
\begin{definition}[Positively controlled invariant set \cite{blanchini2008set}]
Consider the control system $\dot{x}=f(x,u)$ with admissible inputs $u(t)\in\mathcal{U}$. 
A (typically closed) set $\mathcal{C}\subseteq\mathbb{R}^{n}$ is \emph{positively controlled invariant} if, for every $x_0\in\mathcal{C}$, there exists an admissible input signal $u(\cdot)$ 
such that the corresponding trajectory satisfies $x(t;x_0,u)\in\mathcal{C}$ for all $t\ge 0$.
\end{definition}
A safe region can be mathematically defined by equality or inequality constraints on the measured output. Based on this, if the corrective control law ensures that \(\mathcal{C}\) is an invariance region with respect to the system dynamics, then safety is guaranteed or the constraints are satisfied. 
Consider a generalized nonlinear system
\begin{equation}
    \dot{x}=f(x)+g(x)u,\ \  x(0)=x_0,\label{eq.nonlinearsys}
\end{equation}
where \(x\in \mathbb{R}^n\) denotes the system states, \(u \in \mathbb{R}\) denotes the control input, and $f(x)$, $g(x) : \mathbb{R}^n \to \mathbb{R}^n$ are smooth vector fields.  Let the constraints on the states be given by
\begin{equation}
    y_i=h_i(x,\eta_i(t)) \leq 0, \ \ \ \ 1 \leq i \leq q,
    \label{eq.constraints}
\end{equation}
where \(h_i(\cdot):\mathbb{R}^n \times \mathbb{R} \to \mathbb{R}\) is a set of smooth output functions, $y_i$ is the output with a relative degree \(r_i\) and $\eta_i(t) \in \mathbb{R}$ is a time-varying user-defined parameter that characterizes the safe set. The set of all admissible states, denoted as $\mathcal{X}$, can be described as the intersection of the zero–sublevel sets associated with the output functions $h_i(x,\eta_i(t))$. Formally, this is expressed as
\begin{equation}
    \mathcal{X} = \{ x \in \mathbb{R}^n \;|\; h_i(x,\eta_i(t)) \leq 0,\; \forall i = 1,\dots,q \}.
    \label{eq.admissible}
\end{equation}

It is known that $\mathcal{X}$ can be rendered invariant with finite control input only if $r=1$. For higher relative degrees, the approach of~\cite{wolff2004invariance} constructs an admissible set $\mathcal{G}$ that remains invariant through control switching on the boundary $\partial\mathcal{G}$. The set is defined as 
\begin{equation}
    \mathcal{G} = \{ x \;|\; \Phi(x) \leq 0 \}, \qquad 
    \partial \mathcal{G} = \{ x \;|\; \Phi(x)=0 \},
\end{equation}
where $\Phi(x)$ depends on the output $y$, its derivatives $y^{(k)}$ for $k<r$, and a design parameter $\gamma$ ($\gamma=0$ if $r=1$, $\gamma<0$ otherwise). Larger $|\gamma|$ enlarge $\mathcal{G}$ but generally increase the required control effort. \\
The detailed construction of the invariant set is provided in~\cite{wolff2004invariance} and is omitted here for brevity. Based on those results, explicit expressions for $\Phi(x)$ can be derived for systems with low relative degrees, as follows:
\begin{align}
    r=1: & \quad \Phi(x)=y,\\
    r=2: & \quad \Phi(x)=
    \begin{cases}
        y, & \dot{y}\leq 0,\\
        -\tfrac{1}{2\gamma }\dot{y}^2 + y, & \dot{y}>0.
    \end{cases} \label{eq.Phidef}
\end{align}
A sufficient condition for invariance of $\mathcal{G}$~\cite{wolff2004invariance} is that, on $\partial\mathcal{G}$, either
\begin{align}
    y^{(k)}(x) < 0,\quad 1\leq k \leq r-1, \quad \text{or} \quad 
    y^{(r)}(x,u)\leq \gamma,
\end{align}
with the latter enforceable via input–output linearizing control.\\
For multiple constraints, individual sets $\mathcal{G}_i$ can be constructed and combined into a single invariant set
\begin{equation}
    \Phi_{max}(x) = \max_i \Phi_i(x), \qquad 
    \mathcal{G} = \{ x \;|\; \Phi_{max}(x)\leq 0 \},
\end{equation}
which corresponds to the intersection of all $\mathcal{G} _i$ and guarantees admissibility with respect to every constraint~\cite{wolff2007continuous}.

\subsection{System Description and Problem Formulation}
\subsubsection{Robot Dynamics}Consider a rigid manipulator model with \(n\) revolute joints operating in a \(m-\)dimensional task-space
\begin{equation}
M(q)\ddot{q}+C(q,\dot{q})\dot{q}+G(q)=\tau_c+J(q)^Tf_{ext}, \label{eq.ELDynamics}
\end{equation}
where \(q(t),\dot{q}(t),\ddot{q}(t) \in \mathbb{R}^n\) are the joint angle, joint velocity, and joint acceleration vector, respectively. $M(q),C(q,\dot{q}) \in \mathbb{R}^{n \times n}$ denote the inertia matrix and the Coriolis matrix, respectively and $G(q) \in \mathbb{R}^n$ represents the gravity vector.  $\tau_c(t) \in \mathbb{R}^n$ is the control torque provided by the actuators and $f_{ext} \in \mathbb{R}^m$ is the external force experienced by the end-effector due to human-robot interaction. It is assumed that $\|f_{ext}\|<\bar{f}_{ext}$  is known a-priori. $J(q) \in \mathbb{R}^{m \times n}$ is the Geometric Jacobian matrix. The computed torque control input to linearize \eqref{eq.ELDynamics} is  \begin{equation}
    \tau_c=M(q)v+C(q,\dot{q})\dot{q}+G(q)-J(q)^Tf_{ext}, \label{eq.controlinput}
\end{equation}
which results in the closed-loop system as \(\ddot{q}=v.\) Further, we can use \(v=J(q)^{\dag} u-J(q)^{\dag} \dot{J}(q) \dot{q}\)  to obtain the following form:
\begin{equation}
    \ddot{\xi}=u,
\end{equation}
where $\xi \in \mathbb{R}^m$ denotes the end-effector pose and $J(q)^{\dag}$ is the pseudo-inverse of $J(q)$. Now the first objective is to design the control input $u$ such that a desired admittance relation (force to position) is ensured.
\subsubsection{Admittance Relation}The admittance relation \cite{admittance} against the external force (\(f_{ext}\)) is given as follows:
\begin{equation}
    M_a(\ddot{\xi}-\ddot{\xi}_d)+D_a(\dot{\xi}-\dot{\xi}_d)+K_a(\xi-\xi_d)=f_{ext}, \label{eq.admittance}
\end{equation}
where $M_a,D_a,K_a \in \mathbb{R}^{m \times m}$ are the virtual mass, damping, and stiffness matrices, respectively. The predefined desired trajectory in the task-space is denoted as $\xi_d \in \mathbb{R}^m$, i.e., in the absence of any external force, the end-effector should follow $\xi_d.$ Note that $M_a, D_a,$ and $K_a$ are user-defined symmetric positive definite matrices to obtain task-specific compliance. The control input $u$ is designed such that the admittance relation \eqref{eq.admittance} holds, as given below
\begin{equation}
    u=\ddot{\xi}_d-M_a^{-1}(D_a\dot{e}+K_ae-u_c), \quad \label{eq.erroradmittance}
\end{equation}
where $e \triangleq \xi-\xi_d$ and $u_c \in \mathbb{R}^m$ is an auxiliary control input. 
\begin{remark}
   It can be highlighted that using the control input as mentioned in \eqref{eq.erroradmittance}, the admittance relation \eqref{eq.admittance} can be achieved. However, if the desired pose $\xi$ is required to be in a safe set, the same design doesn't cater to that need. To achieve the desired admittance relation along with safety satisfaction following problem is defined.
\end{remark}
\begin{assumption}
    The desired state \( \xi_d \in \mathcal{S}_d\) is within the user-defined safe set \( \mathcal{S}_r \), i.e., $\mathcal{S}_d \subset \mathcal{S}_r$. 
\end{assumption}
\begin{problem}
Given the robot dynamics in \eqref{eq.ELDynamics}, the objective is to design a safe controller $u_c$ such that the admittance relation in \eqref{eq.admittance} is satisfied and pose $\xi$ stays within a user-defined safe state, i.e., \( \xi \in \mathcal{S}_r \) \(\forall t \geq 0\).
\end{problem}
To address this problem, we adopt a two-phase solution. First, a switched model-reference adaptive control framework is formulated to ensure that the reference model trajectory remains within the prescribed safe set $\mathcal{S}$. Subsequently, an error bound $\mathcal{S}_d$ is constructed so that the actual end-effector position is guaranteed to lie within the enlarged safe region $\mathcal{S}_r = \mathcal{S} \oplus \mathcal{S}_d$. This has been illustrated in Figure~\ref{fig:picrepre}
\begin{figure}
    \centering
    \includegraphics[width=\linewidth]{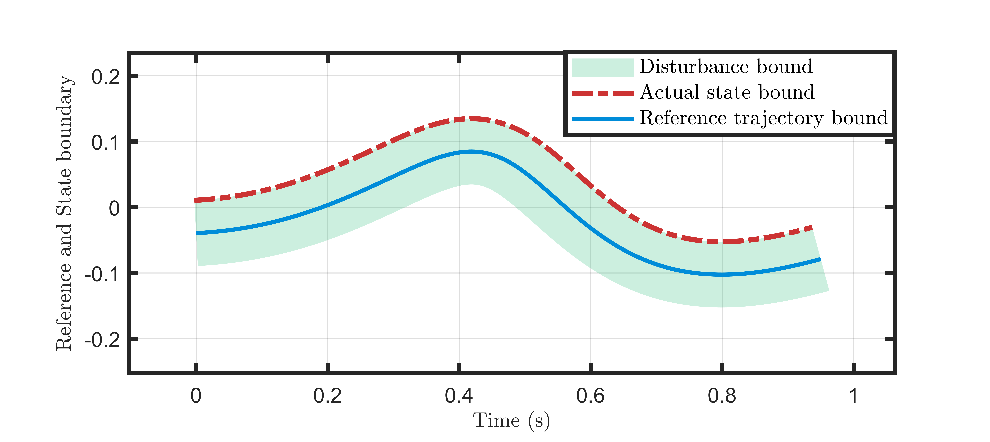}
    \caption{The bound of the set $\mathcal{S}$ is shown in blue solid line, the green region shows the disturbance bound $\mathcal{S}_d$ and the red dot-dashed line shows the actual state boundary, i.e., the boundary of $\mathcal{S}_r.$}
    \label{fig:picrepre}
\end{figure}
\section{Safe Set Construction}
\subsection{Reference Trajectory Shaping} 
Using the control input \eqref{eq.erroradmittance}, the closed-loop system can be expressed in the state-space form as:
\begin{equation}
    \dot{E}=A_a E +B_a u_c, \label{eq.statespace1}
\end{equation}
where \(E=\begin{bmatrix}
    e \\ \dot{e}
\end{bmatrix} \in \mathbb{R}^{2m}\) represents the state vector. The system matrix and the input coupling matrices are defined as
\begin{equation}
    A_a\triangleq\begin{bmatrix}
        0 & I\\ -M^{-1}_aK_a & -M^{-1}_aD_a
    \end{bmatrix},\ B_a\triangleq\begin{bmatrix}
        0 \\ M^{-1}_a
    \end{bmatrix}.\label{eq.sysmatrix}
\end{equation}
The objective here is to design a control such that the safety constraint \(E \in \mathcal{S}\) is satisfied, which can be obtained from $\mathcal{S}_r \ominus\mathcal{S}_d.$
The safety set can be mathematically formulated as
\begin{equation}
    \mathcal{S}=\{\ E \in \mathbb{R}^{2m} | h(E,\eta(t)) \leq 0 \} , \label{eq:hyperplane}
\end{equation}
where \(\eta(t)\) is the user-defined parameter for the formation of the compact set $\mathcal{S}$. An invariance control-based algorithm  is developed in \cite{kimmel2017invariance} to solve this problem where a switching control law is proposed; however, due to the discontinuity in the control law, oscillations are induced in the system trajectory and the control input. To eliminate the undesirable oscillations, our work introduces a switching reference model instead of switching control law, 
which guarantees that \(E \in \mathcal{S}\).
The reference model is chosen to be Piece-wise Affine (PWA) and given as:
\begin{equation}
    \dot{E}_r=A_{p}E_r + B_a f_{ext}, \quad p\in P,\label{eq.refmodels}
\end{equation}
where 
\begin{itemize}
    \item \(E_r \in \mathbb{R}^{2m}\) represents the reference model state vector,
    \item $P=\{ 1, 2\}$ is a finite index set and $\sigma:\mathbb{R}^{2m} \to P$ is the state-dependent switching law that signifies which subsystem is active.
    \item $A_{p} \in \mathbb{R}^{2m \times 2m}$ is the matrix corresponding to the subsystem with the following structure
    \begin{equation}
        A_{p}=\begin{bmatrix}
            0 & I \\ diag[k_{1p},...,k_{mp}] &diag[d_{1p},...d_{mp}]
        \end{bmatrix}.
    \end{equation}
\end{itemize}

\begin{remark}
    The interaction force is assumed to be applied from a single direction. With force along $E_{r_1}$, the sign of $f_{ext_1}$ defines two half-spaces; thus two PWA subsystems suffice. For multi-directional forces, appropriate hyperplanes must be designed accordingly.
\end{remark}
In order to incorporate a switching reference model, two subsystems are defined next.
\subsubsection{Reference Subsystem-1}
For $p=1$, the \textit{Subsystem 1} of the reference model is active and defined as:
\begin{align}
    \dot{E}_r = A_{1} E_r + B_a f_{ext},
\end{align}
which corresponds to a low stiffness profile for better compliance or according to the task. One can also choose $A_1=A_a$. 
\subsubsection{Reference Subsystem-2}Further, to ensure safety, \textit{Subsystem 2} with high stiffness for $p=2$ is defined as 
\begin{align}
    \dot{E}_r=A_2 E_r+B_af_{ext}.
\end{align}
The user-defined parameter matrix $A_2$ is selected such that the following inequality holds:  
\begin{equation}
    2|B_a f_{ext}| \preceq |A_2 \bar{E}_r|, \label{eq.A2cond}
\end{equation}  
where $\bar{E}_r$ is the solution of $h(\bar{E}_r,\bar{\eta}(t))=0$. Here, $\bar{\eta}(t)$ is chosen to satisfy the condition $\bar{\eta}(t)+\bar{D} \leq \eta(t)$ for all $t$, with $\bar{D}$ defined later. Next, define  
\begin{equation}
    \gamma \triangleq A_{2}E_{r} + B_{a}f_{\text{ext}},
\end{equation}
where $A_{2}$ is chosen as a negative definite matrix satisfying~\eqref{eq.A2cond}.  
With this choice, the inequality $\gamma < 0$ holds for all cases.  
Under this assumption, we construct the following function, analogous to~\eqref{eq.Phidef}:

\begin{equation}
    \Phi_1(E_r) =
    \begin{cases}
        h(E_r,\bar{\eta}(t)),  \qquad \quad  \ \ \ \  \ \ \ \ \ \  \dot{h}(E_r,\bar{\eta}(t)) \leq 0, \\
        -\dfrac{1}{2\gamma}\dot{h}^2(E_r,\bar{\eta}(t)) + h(E_r,\bar{\eta}(t)), \quad  \text{otherwise}.
    \end{cases}
\end{equation}  

Using $\Phi_1(E_r)$, the following set is defined:  
\begin{equation}
    \mathcal{G}_1 = \{ E_r \mid \Phi_1(E_r) \leq 0 \}. \label{eq.invariancefunc}
\end{equation}  

The set $\mathcal{G}_1$ is a \emph{positively invariant set} when subsystem~2 is active. Consequently, switching to subsystem~2 guarantees that $h(E_r,\bar{\eta}(t)) \leq 0$ is always satisfied. However, to ensure \eqref{eq:hyperplane}, we need to consider the tracking error dynamics too, which is given in the following section.
\subsection{Error bounds for Safe set}
Choose the control input as
\begin{equation}
    u_c=K_p E+f_{ext}             \label{eq:control}
\end{equation}
where $K_p$ is the control gain with the following assumption.


\begin{assumption}
    $(A_a, B_a)$ pair is controllable. Hence, it is assumed that $\exists$  nominal control gains $K_p^*$ such that $A_a+B_a K_p^*$ is Hurwitz and the following matching condition holds
\begin{equation}
    A_{p}=A_a+B_a K_{p}^* \quad \forall p. \label{matching}
\end{equation}
\end{assumption}
The tracking error ($e_a=E-E_r$) dynamics is given as
\begin{align}
    \dot{e}_a = &(A_a+B_aK_p)E+B_a f_{ext}-A_{p}E_r -B_a f_{ext}\\
                =&A_{p}e_a+B_a\tilde{K}_pE,
\end{align}
where $\tilde{K}_p \triangleq K_p-K_p^*$. Also, from \eqref{eq:hyperplane}, the upper bound on the error can be obtained such that $\|B_a\|\|\tilde{K}_p\|\|E\|\leq D$, where $D$ is some positive constant. Now the worst case scenario for the error dynamics becomes
\begin{equation}
    \dot{e}_a=A_{p}e_a+D\textbf{1}_{2m},      \label{eq:errordyn}
\end{equation}
where $\textbf{1}_{2m}$ is the vector of all-ones. With the initial condition $E_r(0)=E(0)$, the solution of \eqref{eq:errordyn} can be obtained as 
\begin{equation}
    e_a(t)=D\int_0^t e^{A_{p}(t-\tau)} d\tau,
\end{equation}
which can be upper bounded as follows
\begin{equation}
    |e_a|\preceq D \begin{bmatrix}
        \frac{b_{11}\lambda_{12}-b_{12}\lambda_{11}}{k_{11}\Delta_1}\\
        .\\ \frac{b_{i1}\lambda_{i2}-b_{i2}\lambda_{i1}}{k_{i1}\Delta_i}\\
        .\\ \frac{b_{m1}\lambda_{m2}-b_{m2}\lambda_{m1}}{k_{m1}\Delta_m}\\
         \frac{b_{12}-b_{11}}{\Delta_1}\\
         .\\\frac{b_{i2}-b_{i1}}{\Delta_i}\\
         .\\ \frac{b_{m2}-b_{m1}}{\Delta_m}\\
    \end{bmatrix} \label{eq.errorbound}
\end{equation}
where the terms \( b_{i1} = e^{\lambda_{i1} t} - 1 \), \( b_{i2} = e^{\lambda_{i2} t} - 1 \), with eigenvalues \( \lambda_{i1} = \frac{k_{i2} - \Delta_i}{2} \) and \( \lambda_{i2} = \frac{k_{i2} + \Delta_i}{2} \), where \( \Delta_i = \sqrt{k_{i2}^2 + 4k_{i1}} \) for \( i = 1, \dots, m \). 

Taking the derivative of the \( i \)-th row of the upper half of the vector in \eqref{eq.errorbound} with respect to time yields
\begin{equation}
\frac{d}{dt} \left( \frac{b_{i1} \lambda_{i2} - b_{i2} \lambda_{i1}}{k_{i1} \Delta_i} \right)
= \frac{\lambda_{i1} \lambda_{i2} (e^{\lambda_{i1} t} - e^{\lambda_{i2} t})}{k_{i1} \Delta_i} > 0.
\end{equation}

This shows that the first \( m \) components of the bound vector are monotonically increasing functions of \( t \). Therefore, the maximum of these expressions as \( t \to \infty \) is given by:

\begin{equation}
\max_{t > 0} \left\{ \frac{b_{i1} \lambda_{i2} - b_{i2} \lambda_{i1}}{k_{i1} \Delta_i} \right\}
= \lim_{t \to \infty} \left\{ \frac{b_{i1} \lambda_{i2} - b_{i2} \lambda_{i1}}{k_{i1} \Delta_i} \right\}
= -\frac{1}{k_{i1}}.
\end{equation}

For the last \( m \) rows of the vector in Equation~(26), the difference \( \frac{b_{i2} - b_{i1}}{\Delta_i} \) approaches zero both as \( t \to 0 \) and \( t \to \infty \), i.e.,

\begin{equation}
\left. \frac{b_{i2} - b_{i1}}{\Delta_i} \right|_{t = 0}
= \left. \frac{b_{i2} - b_{i1}}{\Delta_i} \right|_{t \to \infty}
= 0.
\end{equation}
From Rolle's theorem, there exists a unique stationary point

\begin{equation}
t_s = \frac{\ln \frac{\lambda_{i2}}{\lambda_{i1}}}{\lambda_{i1} - \lambda_{i2}},
\end{equation}

in the open interval \( (0, \infty) \), at which the first derivative of \( \frac{b_{i2} - b_{i1}}{\Delta_i} \) with respect to \( t \) is zero. At this point, the second derivative of the same function is negative, indicating a local maximum. Thus, the maximum value of \( \frac{b_{i2} - b_{i1}}{\Delta_i} \) occurs at
$t_s$.
Accordingly, the maximum value of the expression \( \frac{b_{i2} - b_{i1}}{\Delta_i} \) can be defined as

\begin{equation}
\beta_i = \max_{t > 0} \left\{ \frac{b_{i2} - b_{i1}}{\Delta_i} \right\} = \frac{e^{\lambda_{i2} t} - e^{\lambda_{i1} t}}{\Delta_i} \bigg|_{t = \frac{\ln \frac{\lambda_{i2}}{\lambda_{i1}}}{\lambda_{i1} - \lambda_{i2}}}.
\end{equation}
Now, defining a set as follows
\begin{equation}
    \mathcal{D}=\left\{ e_a\in \mathbb{R}^{2m}:|e_a|\leq \bar{D}\right\} \label{eq.errorset}
\end{equation}
where $\bar{D}=max(D\begin{bmatrix}
        -\frac{1}{k_{11}} & ...&-\frac{1}{k_{m1}} &\beta_1 &...&\beta_m
    \end{bmatrix} ).$\\
The results established above are stated formally in the theorem below.

\begin{theorem}
    The states of the system \eqref{eq.erroradmittance} tracks the reference trajectory generated by the model \eqref{eq.refmodels} asymptotically while ensuring \(\xi \in \mathcal{S}_r\), if the following indicator function and the control input is implemented.
    \begin{enumerate}
        \item The indicator function is given by
        \begin{equation}
            p=\begin{cases}
                1,\ \ \ \mathcal{G}_1 \oplus \mathcal{D}\in \mathcal{S}_r   \text{ and } \dot{h}(E_r,\bar{\eta}(t)) \leq 0\\
                2,\ \ \ otherwise.
            \end{cases}\label{eq.indicatorfn}
        \end{equation}
        \item  The control input \(u_c\) is given by
        \begin{equation}
             u_c = K_pE + f_{\text{ext}},
        \end{equation}
        where $K_p$ is a dynamically updated gain given by:
        \begin{equation}
            \dot{K}_r = -\Gamma_r B_a^T \mathcal{P} e_a E^T,  \label{eq.Kuplaw}
        \end{equation}
        where $\Gamma_p$ is a  diagonal positive definite gain matrix following the switching rule, and $\mathcal{P}>0$ is a common quadratic Lyapunov function satisfying
        \[
        A_{p}^T \mathcal{P} + \mathcal{P} A_{p} = -I, \quad \forall P.
        \]
  \end{enumerate}
\end{theorem}
\begin{proof}
    Consider the Lyapunov function
    \begin{equation}
        V=\frac{1}{2}e_a^T\mathcal{P}e_a+\frac{1}{2}\sum_{p=1}^2(tr(\tilde{K}_p\Gamma_p^{-1}\tilde{K}_p))
    \end{equation}
    where $\Gamma_p$ is a diagonal positive definite gain matrix.
    The time derivative of the Lyapunov function along \eqref{eq.statespace1} is 
    \small\begin{align}
    \dot{V} &= e^T\bigg(\frac{1}{2}\sum_{p=1}^2 (A_{p}^T\mathcal{P}+\mathcal{P}A_{p})\bigg)e+\sum_{p=1}^2( E^TPB_a\tilde{K}_{p}e_a ) \nonumber \\
    &+\sum_{p=1}^2 tr(\tilde{K}_{p}^T \Gamma_p^{-1}\dot{\tilde{K}}_{p}).\label{vdot}
\end{align}
Choosing $\dot{K}_i$ as given in \eqref{eq.Kuplaw}, the time derivative becomes negative semi-definite as follows
\begin{equation}
    \dot{V}=-\frac{1}{2}e_a^Te_a \leq 0.
\end{equation}
Since, $e_a$ and $\dot{e}_a$ are bounded and $\dot{V}$ is uniformly continuous, using Barbalat's lemma \cite{barbalat}, we can prove the asymptotic convergence of the error. The system state can be upper-bounded as
\begin{align}
    \|E\| \leq \|E_r\|+\|e_a\|.
\end{align}
Now the goal is to choose a switching sequence such that $E_r \in \mathcal{S}_r \ominus \mathcal{D}$ following \eqref{eq.errorset}. 
 Hence, the switching sequence can be
\begin{align}
    p= \begin{cases}
        1, \quad \mathcal{G}_1 \oplus \mathcal{D} \in \mathcal{S}_r \ and \ \ \dot{h}(E_r,\bar{\eta}(t))\leq 0\\
        2, \quad otherwise.
    \end{cases}
\end{align}
This completes the proof.
\end{proof}

\section{Results and Discussion}
\subsection{Simulation Setup}

In the simulation study, we have considered a two-link frictionless manipulator operating in a horizontal plane, subject to task-space constraint. All physical quantities are expressed in SI units. The robot dynamics is modeled following~\cite{craig1987adaptive} and given below:
\begin{align}
\tau_1 &= m_2 l_2^2 \left(\ddot{q}_1 + \ddot{q}_2\right) 
       + m_2 l_1 l_2 c_2 \left(2\ddot{q}_1 + \ddot{q}_2\right) - \tau_{e1} \notag \\
       &\quad + (m_1+m_2) l_1^2 \ddot{q}_1 
       - m_2 l_1 l_2 s_2 \,\dot{q}_2^2 
       - 2 m_2 l_1 l_2 s_2 \,\dot{q}_1 \dot{q}_2, \notag \\[6pt]
\tau_2 &= m_2 l_2^2 \left(\ddot{q}_1 + \ddot{q}_2\right) 
       + m_2 l_1 l_2 c_2 \,\ddot{q}_1 
       + m_2 l_1 l_2 s_2 \,\dot{q}_1^2 
       - \tau_{e2}, \notag
\end{align}
where $\tau_1$ and $\tau_2$ denote the joint torques, 
$q_1, q_2$ are the joint angles, $m_1, m_2$ are the link masses, 
$l_1, l_2$ are the link lengths, and $\tau_{e1}, \tau_{e2}$ are the external torques due to human-robot interaction. For compactness, we define $c_i = \cos(q_i)$, $s_i = \sin(q_i)$, $s_{ij} = \sin(q_i+q_j)$, and $c_{ij} = \cos(q_i+q_j)$. The link parameters are given as $m_1 = 1.5~\mathrm{kg}$, $m_2 = 1.0~\mathrm{kg}$, and $l_1 = l_2 = 0.3~\mathrm{m}$.

\medskip
The  task-space constraints require that the end-effector position $[\xi_1(t),\, \xi_2(t)]^\top$ to satisfy
\begin{align}
|\xi_1(t)| &< k_{c1}(t) = 0.25 + 0.01 \sin(-0.75t), \label{eq.xbound} \\[4pt]
|\xi_2(t)| &< k_{c2} = 0.35, \qquad \forall\, t \geq 0,
\end{align}
with the desired end-effector location specified as
\(
\xi_d = \begin{bmatrix}0.1 \\ 0.2\end{bmatrix}\,\mathrm{m}.
\) The purpose of selecting such a constraint is to highlight that the proposed controller can seamlessly handle time-varying constraints.

The human-robot interaction force applied at the end-effector is defined as
\[
f_{ext}(t) = 
\begin{bmatrix} f_{e1}(t) \\ f_{e2}(t) \end{bmatrix},
\]
with each component described by
\begin{equation}
f_{e_i}(t) =
\begin{cases}
0, & t < 10 \;\;\text{or}\;\; t \geq 21, \\[6pt]
a_i \big(1 - \cos(0.3\pi t)\big), & 10 \leq t < 11, \\[6pt]
2a_i, & 11 \leq t < 20, \\[6pt]
a_i \big(1 + \cos(0.3\pi t)\big), & 20 \leq t < 21,
\end{cases}
\label{eq:external_force}
\end{equation}
where $a_1 = 1$ and $a_2 = 0$.

\begin{figure}
    \centering
    \includegraphics[width=\linewidth]{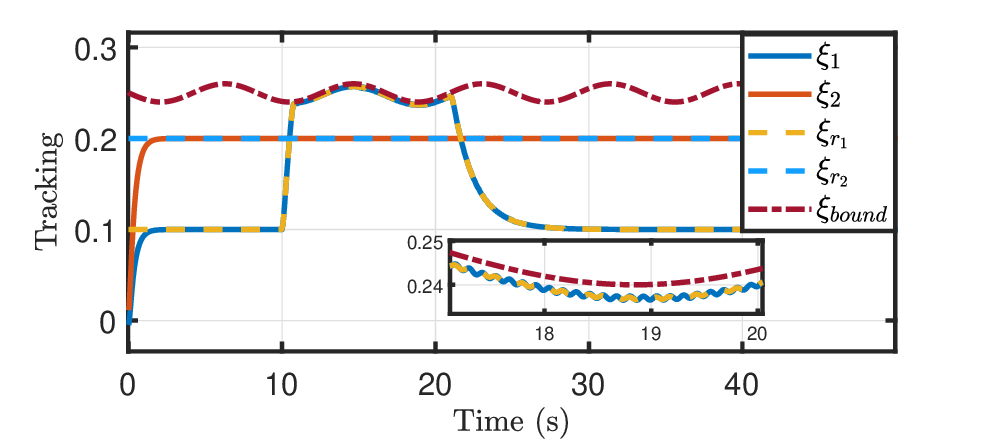}
    \caption{Tracking Response by proposed controller maintaining safety constraint where $\xi_1,\xi_2$ are the end-effector position, $\xi_{r_1}, \xi_{r_2}$ are the reference trajectories, and $\xi_{bound}$ is the bound as given in \eqref{eq.xbound}.}
    \label{fig:mractrack}
\end{figure}
\begin{figure}
    \centering
    \includegraphics[width=\linewidth]{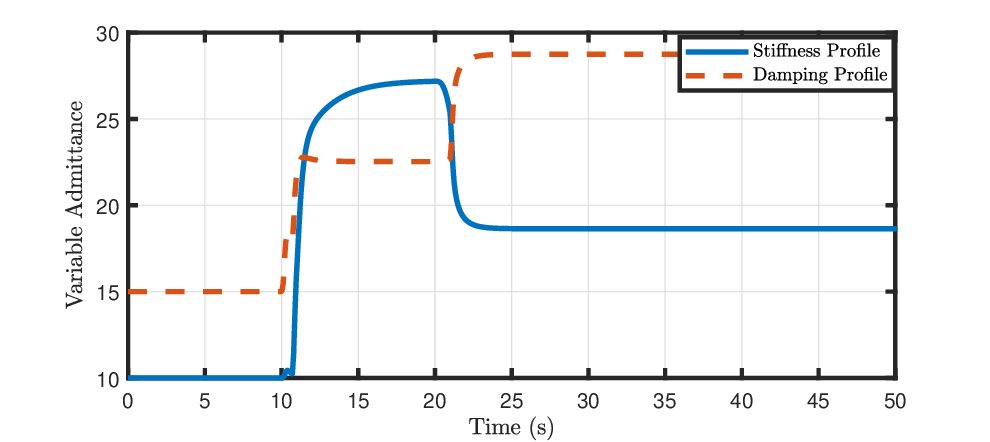}
    \caption{Variable stiffness and damping profiles under switched MRAC design}
    \label{fig:admprof}
\end{figure}
For admittance profile \eqref{eq.admittance} to maintain compliance (i.e. Reference Subsystem-1) is chosen as follows:
\begin{equation}
    A_1=\begin{bmatrix}
        0 & I \\ diag(10,10) & diag (15,15)
    \end{bmatrix},B_a=\begin{bmatrix}
        0\\I
    \end{bmatrix},
\end{equation}
and for Reference Subsystem-2 to enforce safety as
\begin{equation}
    A_2=\begin{bmatrix}
        0 & I \\ diag(40,40) & diag (50,50)
    \end{bmatrix}.
\end{equation}
We further introduce an external disturbance term defined as
\begin{equation}
    d(t) = 0.1 \sin(50t) + 0.05 r_o, \qquad 15 < t < 25,
\end{equation}
where $r_o \in [0,1]$ denotes a random number generated at each time instant.  
As shown in Fig.~\ref{fig:mractrack}, the reference trajectory remains within the safety set prescribed by the indicator function~\eqref{eq.indicatorfn}, where the envelope for disturbance is chosen as $\bar{D} = \tfrac{0.15}{10} = 0.015$.  

Figure~\ref{fig:admprof} illustrates the variable stiffness and damping profile resulting from the switched MRAC design. The controller gain and Lyapunov matrix are selected as
\begin{equation}
    \Gamma_p =
    \begin{bmatrix}
        10 & 0 \\[3pt]
        0 & 8
    \end{bmatrix}, \forall p \qquad
    \mathcal{P} =
    \begin{bmatrix}
        244.4 & 35.05 \\[3pt]
        35.05 & 42.54
    \end{bmatrix}.
\end{equation}

For performance comparison, the proposed method is benchmarked against three state-of-the-art constrained control approaches:  
(i) the ABLF method~\cite{tee2010adaptive},  
(ii) the PPC-based adaptive admittance tracking~\cite{meng2024adaptive}, and  
(iii) invariance control~\cite{hirche}. For this comparison, the task-space constraints are considered time-invariant and defined as
\begin{align}
|\xi_1(t)| &< k_{c1} = 0.35, \\[4pt]
|\xi_2(t)| &< k_{c2} = 0.35, \qquad \forall\, t \geq 0,
\end{align}
while the external interaction force is taken to be the same as specified in~\eqref{eq:external_force}.

\begin{figure}
    \centering
    \includegraphics[width=\linewidth]{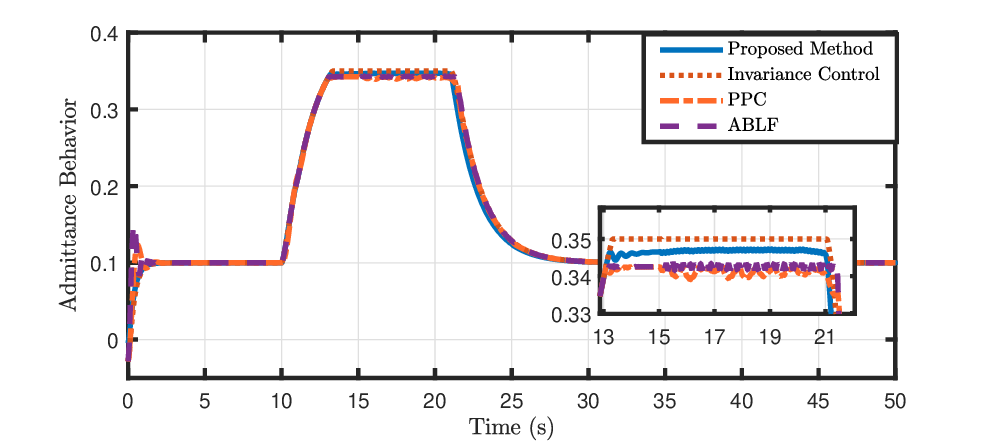}
    \caption{Trajectory along the X-axis: comparison of the proposed controller (blue solid line), invariance controller (red dotted line), PPC (orange dot–dashed line), and ABLF (purple dashed line)}
    \label{fig:trackingcomp}
\end{figure}
\begin{figure}
    \centering
    \includegraphics[width=\linewidth]{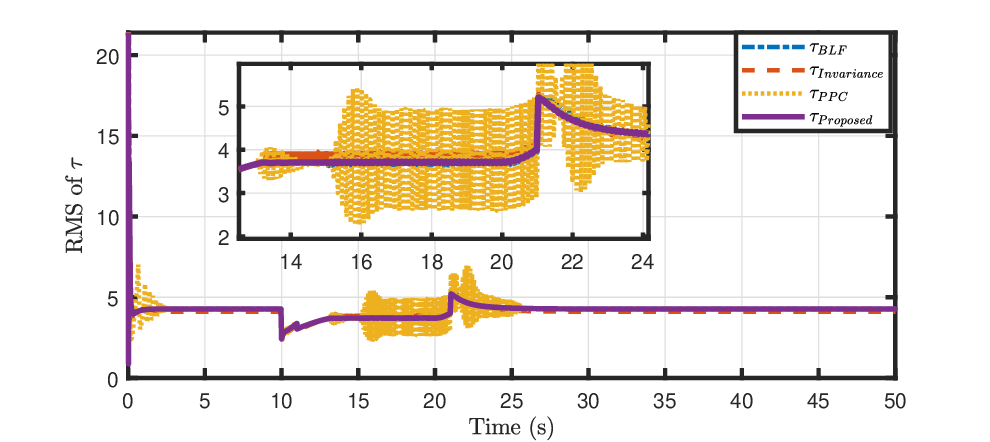}
    \caption{Control input comparison of different controllers}
    \label{fig:contcomp}
\end{figure}
From Fig.~\ref{fig:trackingcomp}, it can be observed that the invariance controller~\cite{kimmel2017invariance} achieves the best admittance behavior under external disturbances while preserving safety. However, this performance comes at the cost of significant control chattering, as highlighted in Fig.~\ref{fig:contcomp}. In particular, the root-mean-square (RMS) values of the control input indicate that both the invariance and PPC controllers suffer from large oscillations, which may hinder practical implementation. In contrast, although the ABLF method produces a smoother control effort comparable to the proposed approach, its tracking accuracy deteriorates considerably in the presence of disturbances.  

\subsection{Comparison with State of the Art}To provide a rigorous performance assessment, Table~\ref{tab:my_label} summarizes quantitative error indices, including the Integral Squared Error (ISE), Integral Absolute Error (IAE), Integral of Time-weighted Squared Error (ITSE), and Integral of Time-weighted Absolute Error (ITAE). These metrics jointly capture both transient and steady-state tracking characteristics. The results clearly demonstrate that the proposed method achieves superior accuracy with substantially reduced control chattering, thereby offering a more balanced trade-off between safety, compliance, and feasibility compared to existing constrained control approaches.

 \begin{table}[ht]
    \centering
    \resizebox{\columnwidth}{!}{%
    \begin{tabular}{|c|c|c|c|c|}
    \hline
     Error metric & Proposed & Invariance\cite{hirche}& ABLF \cite{tee2010adaptive} & PPC \cite{meng2024adaptive} \\
     \hline\hline
     ISE & 0.4889 & 0.5046 & 0.4938 & 0.5014\\
     IAE &  4.9256 & 5.0067&  4.9586& 4.9989 \\
     ITSE & 12.3388 & 12.6828 & 12.4210 & 12.5038 \\
     ITAE & 124.046 & 125.8674& 124.5946 & 125.017\\
     \hline
    \end{tabular}%
    }
    \caption{Error Comparison Metric}
    \label{tab:my_label}
\end{table}
\subsection{Experimental Validation}
To experimentally validate the proposed method, we employ a single-link manipulator testbed, illustrated in Fig.~\ref{fig:hardware}. The setup is developed by modifying a Quanser SRV02 system. It consists of a DC motor actuated through a voltage amplifier, which interfaces with a host computer via Quarc and Simulink. The link position is measured using an encoder, while its velocity is obtained from a tachometer. Sensor signals are acquired by the host computer through a data acquisition board. The controller is implemented in Simulink and executed in real time using hardware-in-the-loop (HIL) read/write blocks. The computed control signal is transmitted to the voltage amplifier, which drives the motor and actuates the link.\\
The single link manipulator system can be modeled as
\begin{equation}
    J_{eq}\dot{\omega}(t)+B_{eq}\omega(t) = A_mV_m(t)+\tau_{ext},  \label{eq.hardwareeq}
\end{equation}
where, $\omega$ and $\theta$ represent the angular velocity and angular position of the link, respectively. The system parameters are identified as follows: the equivalent inertia $J_{eq} = 0.0023~\mathrm{kg \cdot m^2}$, the equivalent damping $B_{eq} = 0.0844~\mathrm{N \cdot m \cdot s/rad}$, and the motor input conversion ratio $A_m = 0.129~\mathrm{N \cdot m/V}$. The term $\tau_{ext} \triangleq f_{ext}l$ denotes the external torque applied to the motor during interaction, where $l=0.1525$ m is the link length.
\begin{figure*}[!t]
    \centering
    \includegraphics[width=0.85\textwidth]{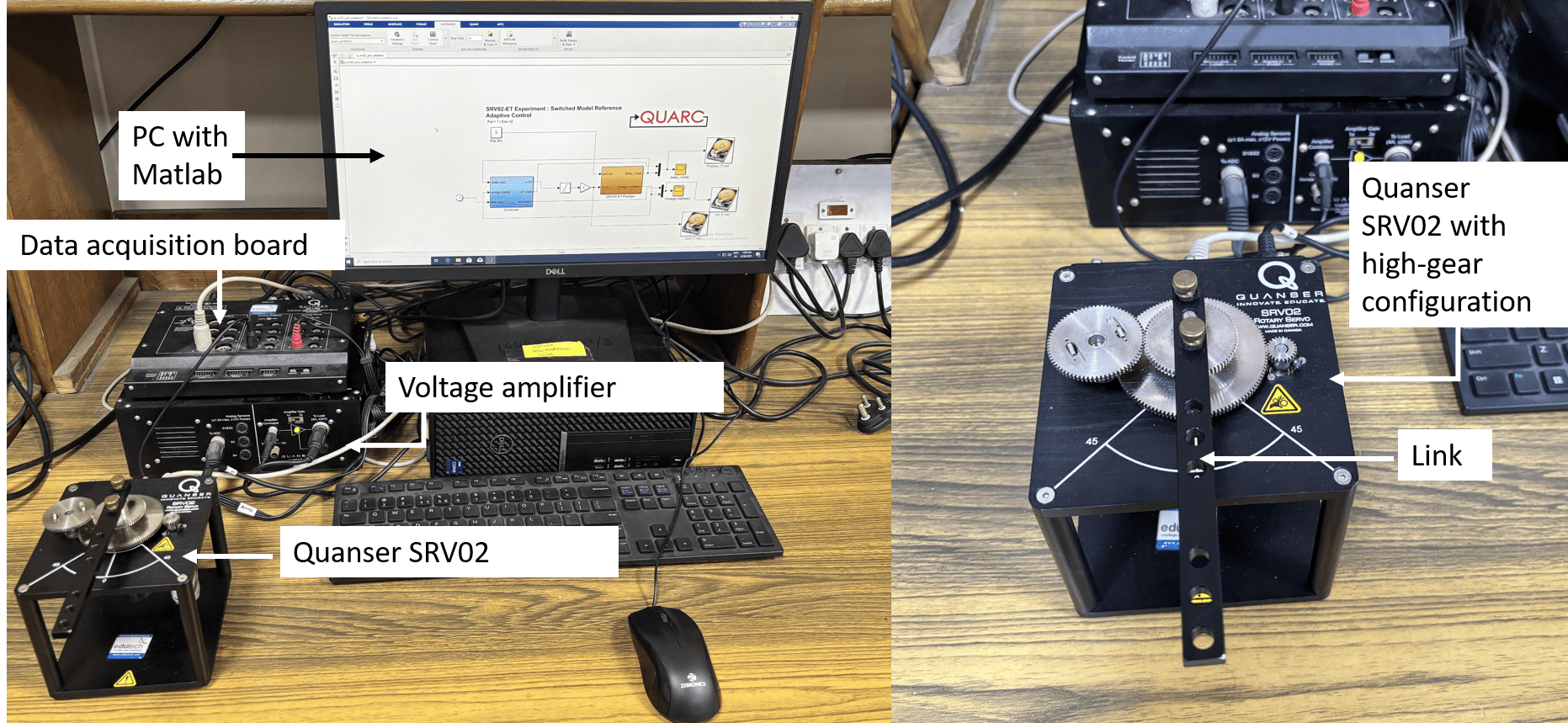}
    \caption{Experimental Hardware}
    \label{fig:hardware}
\end{figure*}
The control objective is to design the input voltage $V_m(t)$ such that the link complies with the external torque $\tau_{ext}$ while ensuring that the safety constraint $\mathcal{S} = \{\theta \in \mathbb{R} : |\theta| \leq 0.25\}$ is satisfied. Since no force measurement device is available in the setup, a residual observer~\cite{magrini2014estimation} is employed and carefully calibrated to estimate the interaction force, as illustrated in Fig.~\ref{fig:hardwareforce}. 
\begin{figure}
    \centering
    \includegraphics[width=0.9
    \linewidth]{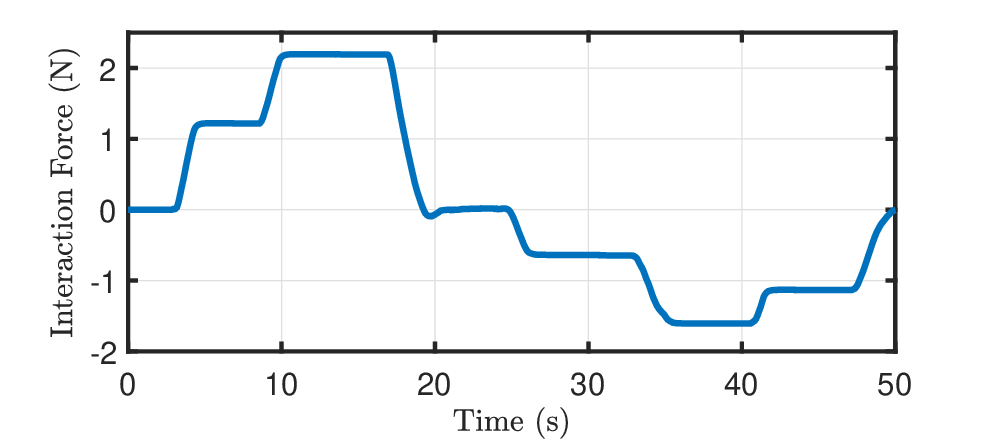}
    \caption{Estimated interaction force ($f_{ext}$)}
    \label{fig:hardwareforce}
\end{figure}
The corresponding admittance model is defined as follows:
\begin{equation}
    A_1=\begin{bmatrix}
        0 & 1\\ -5 & -8
    \end{bmatrix},B_a\begin{bmatrix}
        0 \\ 1
    \end{bmatrix}, A_2=\begin{bmatrix}
        0 & 1\\ -20 & -25
    \end{bmatrix}.
\end{equation}
The noise upper bound is set to $0.05$. Using \eqref{eq.indicatorfn}, the switching surface is designed, and the resulting reference trajectory, along with the actual trajectory, is depicted in Fig.~\ref{fig:harwaretracking}. The actual trajectory remained consistently within the safety set $\mathcal{S}$ with a maximum deviation of less than $0.05$ rad. Compared to baseline controllers, the proposed method exhibited a reduced sensitivity to measurement noise, although at the expense of bounded chattering in the control input, which arises from handling the velocity noise while simultaneously enforcing the safety constraints. Demonstration videos illustrating force estimation and safety achievement using the proposed framework are available online. \footnote{\url{https://shorturl.at/aYlMO}}

\begin{figure}
    \centering
    \includegraphics[width=\linewidth]{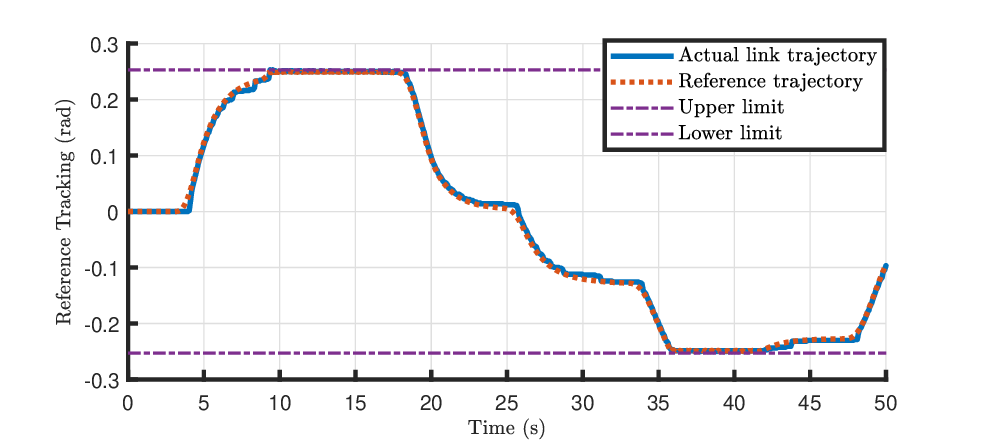}
    \caption{Tracking of reference trajectory by the end-effector maintaining the safety limit}
    \label{fig:harwaretracking}
\end{figure}

\begin{figure}
    \centering
    \includegraphics[width=\linewidth]{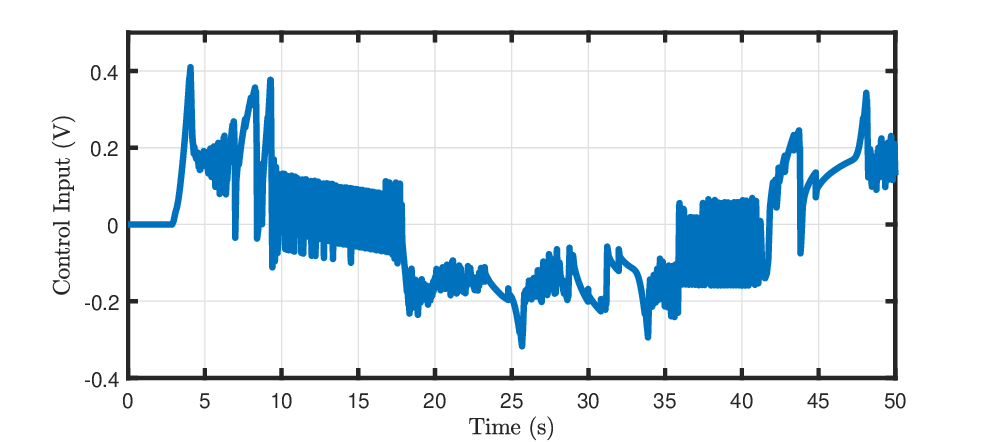}
    \caption{Control input to the single link manipulator}
    \label{fig:controlinput}
\end{figure}

\section{Conclusion}
This paper presented a switched model-reference admittance control framework to ensure both compliance and task-space safety in physical human–robot interaction. By adaptively switching reference models within a variable admittance structure, the proposed approach balances compliance with explicit safety guarantees. An error bound was further incorporated to enhance robustness against disturbances and measurement noise. Theoretical analysis using Lyapunov methods established asymptotic convergence of the switched system, while simulation studies confirmed the method’s ability to satisfy safety constraints with feasible control effort. Comparative evaluations with other controllers demonstrated that the proposed scheme achieves superior tracking accuracy with reduced chattering. Finally, hardware implementation on a single-link manipulator validated the practical feasibility of the approach. Future research will focus on extending the framework to  more complex interaction scenarios, as well as exploring adaptive learning-based switching policies for improved performance under dynamic and uncertain environments.  

\section*{Disclosure statement}
No potential conflict of interest was reported by the author(s).
\section*{Data availability statement}
Data sharing is not applicable to this article as no new data was created or analyzed in this study.

\bibliographystyle{IEEEtran}
\bibliography{root}                           

@article{wolff2004invariance,
  title={Invariance control design for nonlinear control affine systems under hard state constraints},
  author={Wolff, J and Buss, M},
  journal={IFAC Proceedings Volumes},
  volume={37},
  number={13},
  pages={555--560},
  year={2004},
  publisher={Elsevier}
}

@article{craig1987adaptive,
  title={Adaptive control of mechanical manipulators},
  author={Craig, John J and Hsu, Ping and Sastry, S Shankar},
  journal={The International Journal of Robotics Research},
  volume={6},
  number={2},
  pages={16--28},
  year={1987},
  publisher={Sage Publications Sage UK: London, England}
}

@book{blanchini2008set,
  title={Set-theoretic methods in control},
  author={Blanchini, Franco and Miani, Stefano and others},
  volume={78},
  year={2008},
  publisher={Springer}
}

@inproceedings{wolff2007continuous,
  title={Continuous control mode transitions for invariance control of constrained nonlinear systems},
  author={Wolff, Jan and Weber, Carolina and Buss, Martin},
  booktitle={2007 46th IEEE Conference on Decision and Control},
  pages={542--547},
  year={2007},
  organization={IEEE}
}

@article{kimmel2017invariance,
  title={Invariance control for safe human--robot interaction in dynamic environments},
  author={Kimmel, Melanie and Hirche, Sandra},
  journal={IEEE Transactions on Robotics},
  volume={33},
  number={6},
  pages={1327--1342},
  year={2017},
  publisher={IEEE}
}

@article{admittance,
  title={Admittance control for physical human-robot interaction},
  author={Keemink, Arvid QL and van der Kooij, Herman and Stienen, Arno HA},
  journal={The International Journal of Robotics Research},
  volume={37},
  number={11},
  pages={1421--1444},
  year={2018},
  publisher={SAGE Publications Sage UK: London, England}
}

@inproceedings{cbf,
  title={Constrained robot control using control barrier functions},
  author={Rauscher, Manuel and Kimmel, Melanie and Hirche, Sandra},
  booktitle={2016 IEEE/RSJ International Conference on Intelligent Robots and Systems (IROS)},
  pages={279--285},
  year={2016},
  organization={IEEE}
}

@inproceedings{tee2010adaptive,
  title={Adaptive admittance control of a robot manipulator under task space constraint},
  author={Tee, Keng Peng and Yan, Rui and Li, Haizhou},
  booktitle={2010 IEEE International Conference on Robotics and Automation},
  pages={5181--5186},
  year={2010},
  organization={IEEE}
}

@article{meng2024adaptive,
  title={Adaptive admittance tracking control for interactive robot with prescribed performance},
  author={Meng, Qingrui and Lin, Yan},
  journal={Journal of Systems Engineering and Electronics},
  volume={35},
  number={2},
  pages={444--450},
  year={2024},
  publisher={BIAI}
}

@article{Li2018Asymmetric,
  author    = {Li, Z. and Huang, B. and Ajoudani, A. and others},
  title     = {Asymmetric bimanual control of dual-arm exoskeletons for human-cooperative manipulations},
  journal   = {IEEE Transactions on Robotics},
  year      = {2018},
  volume    = {34},
  number    = {1},
  pages     = {264--271},
  doi       = {10.1109/TRO.2017.2776326}
}

@article{Pehlivan2016Minimal,
  author    = {Pehlivan, A. U. and Losey, D. P. and O'Malley, M. K.},
  title     = {Minimal assist-as-needed controller for upper limb robotic rehabilitation},
  journal   = {IEEE Transactions on Robotics},
  year      = {2016},
  volume    = {32},
  number    = {1},
  pages     = {113--124},
  doi       = {10.1109/TRO.2015.2503726}
}

@article{He2017Model,
  author    = {He, W. and Ge, W. and Li, Y. and others},
  title     = {Model identification and control design for a humanoid robot},
  journal   = {IEEE Transactions on Systems, Man, and Cybernetics: Systems},
  year      = {2017},
  volume    = {47},
  number    = {1},
  pages     = {45--57},
  doi       = {10.1109/TSMC.2016.2562505}
}

@inproceedings{magrini2014estimation,
  title={Estimation of contact forces using a virtual force sensor},
  author={Magrini, Emanuele and Flacco, Fabrizio and De Luca, Alessandro},
  booktitle={2014 IEEE/RSJ International Conference on Intelligent Robots and Systems},
  pages={2126--2133},
  year={2014},
  organization={IEEE}
}

@article{abu2020variable,
  title={Variable impedance control and learning—a review},
  author={Abu-Dakka, Fares J and Saveriano, Matteo},
  journal={Frontiers in Robotics and AI},
  volume={7},
  pages={590681},
  year={2020},
  publisher={Frontiers Media SA}
}

@article{song2019tutorial,
  title={A tutorial survey and comparison of impedance control on robotic manipulation},
  author={Song, Peng and Yu, Yueqing and Zhang, Xuping},
  journal={Robotica},
  volume={37},
  number={5},
  pages={801--836},
  year={2019},
  publisher={Cambridge University Press}
}

@article{al2003efficient,
  title={An efficient data-driven fuzzy approach to the motion planning problem of a mobile robot},
  author={Al-Khatib, Mohannad and Saade, Jean J},
  journal={Fuzzy sets and systems},
  volume={134},
  number={1},
  pages={65--82},
  year={2003},
  publisher={Elsevier}
}

@INPROCEEDINGS{chayansafety,
  author={Paul, Chayan Kumar and Shankar Dey, Bhabani and Kar, Indra Narayan},
  booktitle={2024 10th International Conference on Control, Decision and Information Technologies (CoDIT)}, 
  title={Safe Human Robot-Interaction using Switched Model Reference Admittance Control}, 
  year={2024},
  volume={},
  number={},
  pages={1589-1594},
  doi={10.1109/CoDIT62066.2024.10708346}}

@article{Yu2015Human,
  author    = {Yu, H. and Huang, S. and Chen, G. and others},
  title     = {Human-robot interaction control of rehabilitation robots with series elastic actuators},
  journal   = {IEEE Transactions on Robotics},
  year      = {2015},
  volume    = {31},
  number    = {5},
  pages     = {1089--1100},
  doi       = {10.1109/TRO.2015.2457314}
}

@ARTICLE{hirche,
  author={Kimmel, Melanie and Hirche, Sandra},
  journal={IEEE Transactions on Robotics}, 
  title={Invariance Control for Safe Human–Robot Interaction in Dynamic Environments}, 
  year={2017},
  volume={33},
  number={6},
  pages={1327-1342},
  doi={10.1109/TRO.2017.2750697}}

@inproceedings{2008invariance,
  title={Invariance control in robotic applications: Trajectory supervision and haptic rendering},
  author={Scheint, Michael and Wolff, Jan and Buss, Martin},
  booktitle={2008 American Control Conference},
  pages={1436--1442},
  year={2008},
  organization={IEEE}
}

@inproceedings{hogan,
  title={Impedance control: An approach to manipulation},
  author={Hogan, Neville},
  booktitle={American control conference},
  pages={304--313},
  year={1984},
  organization={IEEE}
}

@inproceedings{mariotti2019admittance,
  title={Admittance control for human-robot interaction using an industrial robot equipped with a F/T sensor},
  author={Mariotti, Eleonora and Magrini, Emanuele and De Luca, Alessandro},
  booktitle={2019 International Conference on Robotics and Automation (ICRA)},
  pages={6130--6136},
  year={2019},
  organization={IEEE}
}

@inproceedings{ablf,
  title={Adaptive admittance control of a robot manipulator under task space constraint},
  author={Tee, Keng Peng and Yan, Rui and Li, Haizhou},
  booktitle={2010 IEEE International Conference on Robotics and Automation},
  pages={5181--5186},
  year={2010},
  organization={IEEE}
}

@article{nubert2020safe,
  title={Safe and fast tracking on a robot manipulator: Robust mpc and neural network control},
  author={Nubert, Julian and K{\"o}hler, Johannes and Berenz, Vincent and Allg{\"o}wer, Frank and Trimpe, Sebastian},
  journal={IEEE Robotics and Automation Letters},
  volume={5},
  number={2},
  pages={3050--3057},
  year={2020},
  publisher={IEEE}
}

@article{barbalat,
  title={A {B}arbalat-like lemma with its application to learning control},
  author={Sun, Mingxuan},
  journal={IEEE Transactions on Automatic Control},
  volume={54},
  number={9},
  pages={2222--2225},
  year={2009},
  publisher={IEEE}
}

\end{document}